\documentclass[12pt, article]{l4dc2022}

\title[Adaptive Variants of Optimal Feedback Policies]{Adaptive Variants of Optimal Feedback Policies}
\usepackage{times}
\usepackage[hang,flushmargin]{footmisc}
\usepackage{mathtools}

\usepackage[capitalize]{cleveref}
\crefformat{equation}{(#2#1#3)}
\Crefformat{equation}{Equation~(#2#1#3)}
\Crefname{equation}{Equation}{Eqs.}

\newtheorem{assumption}{Assumption}

\DeclareMathOperator*{\argmin}{arg\,min}

\author{%
 \Name{Brett T. Lopez} \Email{btlopez@ucla.edu}\\
 \addr Verifiable and Control-Theoretic Robotics Laboratory, University of California, Los Angeles, CA%
 \AND
 \Name{Jean-Jacques Slotine} \Email{jjs@mit.edu}\\
 \addr Nonlinear Systems Laboratory, Massachusetts Institute of Technology, MA\\
 Google AI%
}

\begin{document}

\maketitle

\begin{abstract}
   The stable combination of optimal feedback policies with online learning
   is studied in a new control-theoretic framework for uncertain nonlinear systems.
    The framework can be systematically used in transfer learning and sim-to-real applications, where an optimal policy learned for a nominal system needs to remain effective in the presence of significant variations in parameters.
    Given unknown parameters within a bounded range, the resulting adaptive control laws guarantee convergence of the closed-loop system to the state of zero cost. 
    Online adjustment of the learning rate is used as a key stability mechanism, and preserves certainty equivalence when designing optimal policies without assuming uncertainty to be within the control range.
    The approach is illustrated on the familiar mountain car problem, where it yields near-optimal performance despite the presence of parametric model uncertainty.
\end{abstract}


\section{Introduction}

Autonomous decision-making and control have become ubiquitous in many safety-critical systems.
This trend highlights the importance of developing principled algorithms that possess performance guarantees even in the face of uncertainty.
Due to its versatility and generality, optimal control \citep{kirk2004optimal,bertsekas2012dynamic,bryson2018applied,sutton2018reinforcement} is the primary framework for representing and solving difficult decision-making and control problems.
In its purest form, optimal control entails computing a feedback policy that minimizes a cost function given a dynamical model and set of constraints.
While knowing a model is not strictly necessary, e.g., model-free reinforcement learning (RL), any optimal policy will implicitly depend on the underlying dynamics of the system making it susceptible to model uncertainties.
In practice, sensitivity to model perturbations can at best yield suboptimal performance, or at worst result in a catastrophic failure.

Online learning is an effective strategy that reduces sensitivity to model uncertainty while yielding a high-performance, non-conservative feedback policy.
The first common strategy --- known as indirect learning (or system identification) --- generates an explicit model of the underlying dynamics that is used to synthesize a feedback policy.
This approach is generally used in settings that require some form of prediction such as planning or games in addition to optimal control \citep{sutton2018reinforcement,bertsekas2022lessons}. 
Indirect methods, however, require sufficiently rich data to obtain an accurate enough model suitable for control leading to the exploration-exploitation tradeoff.
The second strategy --- referred to as direct learning\footnote{Model-based and model-free RL are synonymous with indirect and direct, respectively.} --- embraces the philosophy of learning just enough about the system to achieve the desired behavior.
Direct methods have a rich history in the controls community, e.g., model reference adaptive control, but have not been fully utilized in optimal control aside from model-free RL \cite{sutton1992reinforcement} and adaptive dynamic programming \citep{murray2002adaptive,vrabie2009neural,lewis2009reinforcement}.
Unfortunately, model-free RL requires extensive offline training and suffers from limited robustness while adaptive dynamic programming needs to iteratively estimate the true cost-to-go online.
A few other approaches have been proposed, e.g., \citep{agarwal2019online,kumar2021rma}, but the complexities and subtleties of combining online learning with nonlinear control often limit their applicability as they rely on linear systems theory or employ ad hoc techniques that do not generalize.

%%% Contributions
\paragraph{Contributions.} 
We develop a new adaptive optimal control framework that utilizes the certainty equivalence principle and online adjustment of the learning rate to guarantee closed-loop stability of near-optimal policies for nonlinear systems with parametric uncertainties. 
The approach consists of combining online learning with optimal value functions and policies computed offline for a family of dynamical systems.
While typically the stability of such a combination cannot be ensured, we can guarantee that the closed-loop system will converge to the state of zero cost by adjusting the learning rate online \citep{lopez2020universal}.
As a result, this work is the first to successfully combine Lyapunov-based learning with optimal control for nonlinear systems.
Two learning algorithms are derived and shown to closely resemble the optimal policy for the well-known mountain car problem despite uncertainties in the dynamical model. 

\paragraph{Scope.}
This work will consider deterministic, time-invariant, continuous-time optimal control of nonlinear systems with parametric uncertainties. 
Extension to other classes of problems is possible and is future work.
More broadly, this work may find uses in model-free reinforcement learning, differential games, sim-to-real, transfer learning, underactuated robotics, or optimal prediction.

\paragraph{Notation.} 
Let $\mathbb{R}_+$ and $\mathbb{R}_{>0}$ denote the set of positive and strictly positive reals, respectively.
The shorthand notation for a function $T$ parameterized by a vector $a$ with vector argument $s$ will be $T_a(s) \triangleq T(s;a)$.
The partial differentiation with respect to variable $x \in \mathbb{R}^n$ of function $M(x,y)$ will be $\nabla_x M(x,y) = \partial M / \partial x \in \mathbb{R}^{n}$.
The subscript for $\nabla$ will be omitted when it is clear which variable the differentiation is with respect to. 
The desired terminal state will be denoted as $x_d$. 

%%%%%%%%%%%%%%%%%%%%%%%%%%%%%%%%%
\section{Optimal Control Review}

Consider the deterministic, infinite-horizon optimal control problems of the form
\begin{equation}
\label{eq:ocp}
    \begin{aligned}
        V^*(x(t)) & =\underset{\pi \, \in \, \Pi}{\mathrm{min}} ~  \int \limits_{t}^\infty \ell (x(\tau),\pi(x(\tau)))\,d\tau \\
        \dot{x} & = F_\theta(x,\pi(x))
    \end{aligned}
\end{equation}
with state $x \in \mathbb{R}^n$, feedback policy $\pi : \mathbb{R}^n \rightarrow \mathbb{R}^m$, partially known dynamics $F_\theta: \mathbb{R}^n \times \mathbb{R}^m \times \mathbb{R}^p \rightarrow \mathbb{R}^{n}$ parameterized by unknown coefficients $\theta \in \mathbb{R}^p$, stage cost $\ell : \mathbb{R}^n \times \mathbb{R}^m \rightarrow \mathbb{R}_+$ where $\ell(x,\pi(x)) = 0 \iff x = x_d$, and control constraint set $\Pi$.
The optimal value function $V^* : \mathbb{R}^n \rightarrow \mathbb{R}_+$ is the cost-to-go from the initial state $x(t)$ if the optimal policy $\pi^*(x)$ is executed indefinitely.
The initial state can be represented as a function of time, i.e., $x(t)$, because the optimal control problem \cref{eq:ocp} possess an important time-invariant property where the initial time can be arbitrarily denoted as $t$.
This property will be useful in analyzing the stability of the proposed approach.
Note that a value function can also be defined for a suboptimal stable policy $\pi(x)$ and satisfies $V^*(x) \leq V^\pi(x) < \infty$ where the last inequality follows from $\pi(x)$ being a stabilizing controller\footnote{A non-finite value function would indicate either $x \nrightarrow x_d$ or $x \rightarrow x_d$ but ``slow enough" that infinite cost is accumulated. Imposing $V^\pi(x)$ be finite eliminates both scenarios.}.
The following assumptions are made about the dynamics in \cref{eq:ocp}.

\begin{assumption}
\label{assumption:dynamics}
The dynamics $F_\theta(\cdot)$ can be decomposed into known and unknown dynamics with $f : \mathbb{R}^n \times \mathbb{R}^m \rightarrow \mathbb{R}^n$ representing the known part.
\end{assumption}

\begin{assumption}
\label{assumption:uncertainty}
The uncertain part of the nonlinear dynamics in \cref{eq:ocp} can be expressed as a linear combination of known basis functions $\Delta : \mathbb{R}^n \rightarrow \mathbb{R}^{p\times n}$ and unknown parameters $\theta \in \mathbb{R}^p$ that belong to a closed convex set $\Theta$.
\end{assumption}

\begin{remark}
\cref{assumption:uncertainty} is not overly restrictive as systems with non-parametric or nonlinearly parameterized uncertainties can be converted into a linear weighting of handpicked or learned basis functions \citep{o2021meta, richards2021adaptive} 
\end{remark}

A common formulation for optimal control of dynamical systems with bounded parametric uncertainties is the following minimax optimization
\begin{equation}
\label{eq:minimax}
    \begin{aligned}
        V^*(x(t)) & = \underset{\pi \, \in \, \Pi}{\mathrm{min}} ~ \underset{\theta \, \in \, \Theta}{\mathrm{max}} ~  \int \limits_t^\infty \ell (x(\tau),\pi(x(\tau)))\,d\tau \\
        \dot{x} & = f(x,\pi(x)) - \Delta(x)^\top \theta.
    \end{aligned}
\end{equation}
The optimal value function for \cref{eq:minimax} can be obtained by solving the Hamilton-Jacobi-Isaacs equation
\begin{equation}
\label{eq:hjbi}
   \underset{\pi \, \in \, \Pi}{\mathrm{min}} \, \underset{\theta \, \in \, \Theta}{\mathrm{max}} \left\{ \ell(x,\pi(x)) +  \nabla_x V^*(x)^\top \left( f(x,\pi(x)) - \Delta(x)^\top \theta \right) \right\} = 0.
\end{equation}
Conceptually, solving \cref{eq:hjbi} is equivalent to computing a \emph{single policy for all possible dynamical models}, i.e., a robust control strategy, which will inherently perform worse than its optimal counterpart if $\theta$ where known. 
We instead propose a certainty equivalence approach that uses a \emph{family} of value functions $V^*_\theta(x)$ and policies $\pi^*_\theta(x)$, parameterized by the unknown model parameter vector $\theta$, that satisfy the Hamilton-Jacobi-Bellman (HJB) equation
\begin{gather}
\label{eq:hjb}
    \underset{\pi_\theta \, \in \, \Pi}{\mathrm{min}} \left\{ \ell(x,\pi_\theta(x)) +  \nabla_x V_\theta^* (x)^\top \left( f(x,\pi_\theta(x)) - \Delta(x)^\top \theta \right) \right\} = 0, ~~ \text{for each} ~\theta \in \Theta.
\end{gather}
The difference between \cref{eq:hjbi} and \cref{eq:hjb} is quite significant from a theoretical and practical point of view. 
The policy generated by solving \cref{eq:hjbi} is the best policy given the worst possible realization of the model uncertainty; in essence a ``one-size-fits-all" approach.
Alternatively, the policy satisfying \cref{eq:hjb} will be optimal if $\theta$ is known.
From the perspective of online learning, \cref{eq:hjb} allows one to employ the certainty equivalence principle to obtain a near-optimal policy with the current model estimate.
The challenge then becomes designing the learning algorithm since combining a stable estimator and feedback policy does not necessarily yield a stable closed-loop if the system is nonlinear \citep{krstic1995nonlinear}.
Moreover, the model-parameterized value functions present additional challenges as learning transients can also lead to unstable control.
Next, two online learning algorithms will be proposed that can be stably combined with model-parameterized value functions and policies.

%%%%%%%%%%%%%%%%%%%%%%%%%%%%%%%%%
\section{Theory}
\subsection{Preliminaries}

Before presenting the main theorems, we first recall a simple but fundamental result due to \citep[p.~387]{kalman1960control}, see also \citep[p.~425-427]{luenberger1979introduction}: an optimal value function $V^*(x)$ is also a control Lyapunov function.
This fact is central to our approach.
\begin{proposition}
\label{prop:clf}
An optimal value function $\,V^*(x)\,$ is also a control Lyapunov function.
\end{proposition}
\begin{proof}
Recall $\,V^*(x(t)) = \int_t^\infty \ell (x(\tau),\pi(x(\tau)))\,d\tau > 0\,$ for all $x \neq x_d$ and $V^*(x_d) = 0$.
Since $\dot{V}^*(x) = - \ell(x,\pi^*(x)) \ < \ 0 \ $ for all $\ x \neq x_d \ $, then $ \ V^*(x) \ $ is a control Lyapunov function.
\end{proof}

Proposition~\ref{prop:clf} can be extended to suboptimal policies $\pi(x)$, e.g., policies which approximate $\pi^*(x)$, as long as one can show that the associated value function $V^\pi(x)$ is finite over the operating domain.
Hence, once a valid value function $V^\pi(x)$ is known, whether optimal or not, one can conclude the closed-loop system converges to $x_d$ with policy $\pi(x)$ .

$ \ $

The proposed approach requires the following two differentiabilty assumptions. 

\begin{assumption}
\label{assumption:ell}
The stage cost $\ell(\cdot)$ is continuously differentiable.
\end{assumption}

\begin{assumption}
\label{assumption:pi}
The optimal value function $V^*_\theta(x)$ and policy $\pi^*_{\theta}(x)$ are continuously differentiable.
\end{assumption}

Assumption~\ref{assumption:ell} is easy to ensure by appropriate selection of $\ell(\cdot)$. 
Assumption~\ref{assumption:pi} may seem more restrictive as it excludes discontinuous optimal policies, e.g., bang-bang control. 
However, an optimal policy is often itself the result of 
as differentiable computing pipeline, or otherwise can be smoothed by appropriate selection of a continuously differentiable function.
For instance, if the optimal policy is restricted to take discrete values, e.g., $\pi^*(x) \in \Pi = \{ \pi^*(x): \{-1,\,0\,1\},~ \forall x\}$, then functions like the saturation function, logistic function, etc.~can be used to make $\pi^*(x)$ continuously differentiable with negligible performance degradation given a reasonable choice of parameters. 
Similar smoothing can be applied to the optimal value function.

We will make use of the Bregman divergence operator in our online learning algorithm to impose physical consistency \citep{wensing2017linear,lee2018natural} or sparsity based on past trajectories \citep{ghai2020exponentiated,boffi2021higher} of the parameter estimates. 

\begin{definition}[Bregman Divergence]
Let $\psi(\cdot)$ be a strictly convex, continuously differentiable function on a closed convex set. 
The \emph{Bregman divergence} associated with $\psi(\cdot)$ is given by 
\begin{equation}
\label{eq:bregman}
    \mathrm{d}_\psi ( y ~\|~ x) = \psi(y) - \psi(x) - (y-x)^\top \nabla \psi(x),
\end{equation}
with a time-derivative that satisfies \, $\dot{\mathrm{d}}_\psi (y~\|~x) = (x-y)^\top \nabla^2 \psi(x) \, \dot{x}$.
\end{definition}

\subsection{Direct Learning}
We now state the first main technical result of this work.

\begin{theorem}
Let $\upsilon(\cdot)$ be any strictly-increasing, strictly-positive scalar function and $\psi(\cdot)$ be a continuously differentiable, strictly convex function on a closed convex set.
If a value function $V^*_\theta({x})$ and policy $\pi^*_{{\theta}}(x)$ can be computed for each $\theta \in \Theta$ then the closed-loop system asymptotically converges to the state of zero cost $x_d$ with the policy $\pi^*_{\hat{\theta}}(x)$ and learning algorithm
\begin{subequations}
\label{eq:direct}
    \begin{align}
        \dot{\hat{\theta}} & = - \gamma\, \upsilon(\rho) [\nabla^2\psi(\hat{\theta})]^{-1} \Delta(x) \nabla_{x} V^*_{\hat{\theta}}(x), \label{eq:dot_theta_V} \\
        \dot{\rho} & = - \frac{\upsilon(\rho)}{\nabla \upsilon(\rho)} \sum\limits_{i=1}^p \left[ \frac{1}{V^*_{\hat{\theta}}(x) + \eta} \, \nabla_{\hat{\theta}_i} V^*_{\hat{\theta}}(x) \right] \, \dot{\hat{\theta}}_i, \label{eq:dot_rho_V}
    \end{align}
\end{subequations}
where $\gamma \in \mathbb{R}_{>0}$ is the learning rate and $\eta \in \mathbb{R}_{>0}$.
\label{thm:direct}
\end{theorem}

\begin{proof}
Consider the Lyapunov-like function
\begin{equation}\label{eq:V_c_V}
    V_c(t) = \upsilon(\rho) (V^*_{\hat{\theta}}(x)+\eta) + \tfrac{1}{\gamma} \mathrm{d}_\psi(\theta~\|~\hat{\theta}),
\end{equation}
where $0 < \eta < \infty$ and $\tilde{\theta} \triangleq \hat{\theta} - \theta$.
Differentiating \cref{eq:V_c_V} along the unknown dynamics, 
\begin{equation*}
    \begin{aligned}
        \dot{V}_c (t) &  = \upsilon(\rho) \dot{V}^*_{\hat{\theta}}(x) + \dot{\rho}\, \nabla \upsilon(\rho) (V^*_{\hat{\theta}}(x)+\eta) + \tfrac{1}{\gamma} \tilde{\theta}^\top \nabla^2 \psi({\hat{\theta}})\,\dot{\hat{\theta}} \\
        & = \upsilon(\rho) \Big[ \nabla_x V^*_{\hat{\theta}}(x)^\top \left( f(x,\pi^*_{\hat{\theta}}(x)) - \Delta(x)^\top \theta \right)  + \sum_{i=1}^p \nabla_{\hat{\theta_i}} V^*_{\hat{\theta}}(x) \, \dot{\hat{\theta}}_i \Big] \\
        & \hphantom{=} + \dot{\rho}\, \nabla \upsilon(\rho) (V^*_{\hat{\theta}}(x)+\eta) + \tfrac{1}{\gamma} \tilde{\theta}^\top \nabla^2 \psi({\hat{\theta}})\,\dot{\hat{\theta}}.
    \end{aligned}
\end{equation*}
Note the presence of the sign-indefinite terms $\, \sum_{i=1}^p\nabla_{\hat{\theta}_i} V^*_{\hat{\theta}}(x) \, \tfrac{d}{dt}{\hat{\theta}}_i \, $ which appear because the value function is parameterized by $\hat{\theta}$.
Using \cref{eq:hjb} and the definition $\theta = \hat{\theta} - \tilde{\theta}$, 
\begin{equation*}
    \begin{aligned}
        \dot{V}_c (t) = & ~ \upsilon(\rho) \nabla_x V^*_{\hat{\theta}}(x)^\top \left( f(x,\pi^*_{\hat{\theta}}(x)) - \Delta(x)^\top \hat{\theta} \right)  + \upsilon(\rho) \nabla_x V^*_{\hat{\theta}}(x)^\top \Delta(x)^\top \tilde{\theta} \\
        & + \sum_{i=1}^p \nabla_{\hat{\theta_i}} V^*_{\hat{\theta}}(x) \dot{\hat{\theta}}_i  + \dot{\rho}\, \nabla \upsilon(\rho) (V^*_{\hat{\theta}}(x)+\eta) + \tfrac{1}{\gamma} \tilde{\theta}^\top \nabla^2 \psi({\hat{\theta}})\,\dot{\hat{\theta}} \\
        = & - \upsilon(\rho) \ell(x,\pi^*_{\hat{\theta}}(x)) + \upsilon(\rho) \nabla_x V^*_{\hat{\theta}}(x)^\top \Delta(x)^\top \tilde{\theta} + \sum_{i=1}^p \nabla_{\hat{\theta_i}} V^*_{\hat{\theta}}(x) \dot{\hat{\theta}}_i \\\
        & + \dot{\rho}\, \nabla \upsilon(\rho) (V^*_{\hat{\theta}}(x)+\eta) + \tfrac{1}{\gamma} \tilde{\theta}^\top \nabla^2 \psi({\hat{\theta}})\,\dot{\hat{\theta}}.
    \end{aligned}
\end{equation*}
Using~\cref{eq:dot_theta_V,eq:dot_rho_V} yields $\, \dot{V}_c(t) =  - \upsilon(\rho) \ell(x,\pi^*_{\hat{\theta}}(x)) \leq 0 \, $ which implies that $\, \tilde{\theta} \,$ and the product $\, \upsilon(\rho) (V^*_{\hat{\theta}}(x) + \eta) \,$ are bounded.
Since $\, V^*_{\hat{\theta}}(x) > 0 \, $ for all $\, x \neq x_d \, $ and $\, \upsilon(\rho) > 0$ uniformly, then both $\, V^*_{\hat{\theta}}(x)\, $ and $\, \upsilon(\rho)\, $ are bounded for all $\, x \neq x_d$.
When $\, x = x_d \,$ the learning rate is zero, i.e., $\, \tfrac{d}{dt}{{\hat{\theta}}} = 0\,$, which implies $\, \dot{\rho} = 0\, $ from \cref{eq:dot_rho_V} so $\, \upsilon(\rho)\, $ remains bounded.
Hence, $\, V^*_{\hat{\theta}}(x)\, $ is bounded because $\, \eta$, $\, \upsilon(\rho)$, and $\, \upsilon(\rho) (V^*_{\hat{\theta}}(x) + \eta) \, $ are bounded.
Furthermore, boundedness of $\, V^*_{\hat{\theta}}(x)\, $ implies that $x$ is also bounded. 
Differentiating $\, \upsilon(\rho) \ell(x,\pi^*_{\hat{\theta}}(x))  \,$ and utilizing \cref{eq:dot_rho_V}, 
\begin{equation*}
    \begin{aligned}
        \frac{d}{dt} [\upsilon(\rho) \ell(x,\pi^*_{\hat{\theta}}(x))]  
          =  \upsilon(\rho) \dot{\ell}(x,\pi^*_{\hat{\theta}}(x)) - \upsilon(\rho) \ell(x,\pi^*_{\hat{\theta}}(x)) \sum\limits_{i=1}^p \left[ \frac{1}{V^*_{\hat{\theta}}(x) + \eta} \, \nabla_{\hat{\theta}_i} V^*_{\hat{\theta}}(x) \right] \, \dot{\hat{\theta}}_i    \\
    \end{aligned}
\end{equation*}
which is bounded by smoothness of $\, V^*_{\hat{\theta}}(x),~ \ell(x,\pi_{\hat{\theta}}^*(x))\,$ and boundedness of $x$, $\upsilon(\rho)$.
Hence, $\upsilon(\rho) \ell_{\hat{\theta}}(x,\pi^*(x)) \, $ is uniformly continuous.
Noting $V_c(0)$ is initially bounded, integrating $\tfrac{d}{dt}{V}_c(t)$ yields $\int_0^\infty \upsilon( \rho(\tau) ) \ell_{\hat{\theta}}(x(\tau),\pi^*(x(\tau))) \, d \tau < \infty$
so by Barbalat's lemma $\, \upsilon (\rho) \ell_{\hat{\theta}}(x,\pi^*(x)) \rightarrow 0$.
Since $\, \upsilon (\rho) > 0 \,$ uniformly and $\, \ell_{\hat{\theta}}(x,\pi^*(x)) = 0 \iff x = x_d \, $  then $ \, x \rightarrow x_d \, $ as $t\rightarrow +\infty$.
\end{proof}

\cref{thm:direct} is highly versatile. 
The direct learning algorithm \cref{eq:direct} \emph{guarantees} the system will converge to the state of zero cost despite nonlinear dynamics and {without} learning the true underlying model.
This result is an immediate consequence of using direct adaptive control formulation.
The intuition behind \cref{thm:direct} is that the function $\upsilon(\rho)$ actively adjusts the learning rate online thereby leading to an \emph{effective} learning rate $\, \gamma \, \upsilon(\rho)$.
Conceptually, the online adjustment is needed to ensure the learning transients, i.e., $\tfrac{d}{dt}\hat{\theta}$, does not destabilize the closed-loop system.
Inspecting \cref{eq:direct}, the effective learning rate is slowed when the learning transients is destabilizing while the converse is true when the learning transients is stabilizing.
The online adjustment allows us to employ the certainty equivalence principle, i.e., select the optimal policy with the current parameter estimate, without concerns of instability.

The direct learning law has an interesting connection with Lagrange multipliers (or costates of \cref{eq:ocp}).
Denoting $\, \lambda^* \in \mathbb{R}^n \, $ as the Lagrange multipliers for the optimal solution of \cref{eq:ocp}, one can trivially show $\, \nabla_x V^*_{{\theta}}(x) = \lambda^*$.
The adaptation law \cref{eq:direct} can thus be interpreted as updating the parameters in the direction of lower cost-to-go along the vector field $\Delta(x)$.
This relation provides a potentially attractive alternative to computing the gradient of $\,V^*(x)\,$ numerically, which may be cumbersome for high-dimensional systems.
It may also lead to a new adaptive optimal control paradigm based on the Hamiltonian function $\mathcal{H}_\theta (x,\lambda) \triangleq \ell(x,\pi_\theta(x)) + \lambda^\top F_{{\theta}}(x,\pi_{{\theta}}(x))$ where 
\begin{equation*}
    \begin{aligned}
        \dot{x} &= \nabla_\lambda \mathcal{H}_\theta(x,\lambda) \\
        \dot{\lambda} &= - \nabla_x \mathcal{H}_\theta(x,\lambda)
    \end{aligned}
\end{equation*}
which, in combination with the condition $\pi^*_{{\theta}}(x) = \argmin_{\pi\in\Pi} \mathcal{H}_{{\theta}}(x,\lambda)$, can be easier to solve than the HJB equation as the above is just a set of ordinary differential equations.
Future work will investigate this dual framework in more detail.

\subsection{Composite Learning}
The direct learning algorithm in \cref{eq:direct} uses $\nabla_x V_{\hat{\theta}}^*(x)$ as an error signal to update the parameter estimates.
While stable control and learning is the ultimate goal, better performance can be obtained by adding a state predictor to the adaptation law.
A state predictor is simply a form of model estimation that, when combined with direct learning, yields smoother learning transients.
We refer to the combination of direct learning and state prediction as \emph{composite learning}.
\begin{definition}[State Predictor]
    \label{def:ctsp}
    The \emph{state predictor} is $\,\varepsilon_{\hat{\theta}}(x) \triangleq \tfrac{d}{dt}({x} - x_{\hat{\theta}})\,$ where $\,\tfrac{d}{dt}x_{\hat{\theta}} \,$ is the instantaneous state speed with the current parameter estimate, i.e., $ \,\tfrac{d}{dt}x_{\hat{\theta}} = F_{\hat{\theta}}(x,\pi^*_{\hat{\theta}}(x))$.
\end{definition}

The benefit of using Definition~\ref{def:ctsp} as opposed to other predictors is the relation between $\varepsilon_{\hat{\theta}}(x)$ and $\tilde{\theta}$, namely $\varepsilon_{\hat{\theta}}(x) = \Delta(x)^\top \tilde{\theta}$ which can be arrived at using the definition of $\tfrac{d}{dt}x$ and $\tfrac{d}{dt}x_{\hat{\theta}}$.
This identity will be important for showing stability of the composite learning algorithm.

\begin{remark}
    If $\tfrac{d}{dt}{x}$ if not directly available,
    one can use in its place a filtered velocity while preserving
    linear parametrization.
    For instance, if we let $\tfrac{d}{dt}\hat{x} = \beta (x - \hat{x})$ where $\beta \in \mathbb{R}_{>0}$ then it is straightforward to show the equivalence 
    \begin{equation*}
        \dot{\hat{x}} = \hat{f}(x,\pi^*(x)) - \hat{\Delta}(x)^\top \theta
    \end{equation*}
    where $\hat{f}(x,\pi^*(x))$ and $\hat{\Delta}(x)$ are the filtered version of $f(x,\pi^*(x))$ and $\Delta(x)$, respectively.
    If identical filtering is done for $\tfrac{d}{dt}{x}_{\hat{\theta}}$, then one can let $\varepsilon_{\hat{\theta}}(x) = \tfrac{d}{dt}(\hat{x} - \hat{x}_{\hat{\theta}})$ which only depends on filtered quantities that are immediately available. 
    Moreover, we still maintain the key property $\varepsilon_{\hat{\theta}}(x) = \hat{\Delta}(x)^\top \tilde{\theta}$.
    This modification is just a practical one, as the proposed composite learning algorithm will yield a stable adaptive policy whether actual or filtered state velocity is used.
\end{remark}

We now present the second main technical result of this work.

\begin{theorem}
Let $\upsilon(\cdot)$ be any strictly-increasing, strictly-positive scalar function and $\psi(\cdot)$ be a continuously differentiable, strictly convex function on a closed convex set.
If a value function $V^*_\theta({x})$ and policy $\pi^*_{{\theta}}(x)$ can be computed for each $\theta \in \Theta$ then the closed-loop system asymptotically converges to the state of zero cost $x_d$ with the policy $\pi^*_{\hat{\theta}}(x)$ and learning algorithm
\begin{subequations}
\label{eq:indirect}
    \begin{align}
        \dot{\hat{\theta}} & = - \gamma\, \upsilon(\rho) [\nabla^2\psi(\hat{\theta})]^{-1} \Delta(x) \nabla_{x} V^*_{\hat{\theta}}(x) - \alpha \, [\nabla^2\psi(\hat{\theta})]^{-1} \Delta(x) \varepsilon_{\hat{\theta}}(x), \label{eq:dot_theta_indirect} \\
        \dot{\rho} & = - \frac{\upsilon(\rho)}{\nabla \upsilon(\rho)} \sum\limits_{i=1}^p \left[ \frac{1}{V^*_{\hat{\theta}}(x) + \eta} \, \nabla_{\hat{\theta}_i} V^*_{\hat{\theta}}(x) \right] \, \dot{\hat{\theta}}_i, \label{eq:dot_rho_indirect}
    \end{align}
\end{subequations}
where $\varepsilon_{\hat{\theta}}(x)$ is the state predictor, $\gamma \in \mathbb{R}_{>0}$ is the direct learning rate, $\alpha \in \mathbb{R}_{>0}$ is the prediction learning rate, and $\eta \in \mathbb{R}_{>0}$.
\label{thm:indirect}
\end{theorem}
\begin{proof}
% The proof largely follows that of \cref{thm:direct}. 
Using the same Lyapunov-like function \cref{eq:V_c_V}, one can show
\begin{equation*}
    \begin{aligned}
        \dot{V}_c(t) & = - \upsilon(\rho) \ell(x,\pi^*_{\hat{\theta}}(x)) - \tfrac{\alpha}{\gamma} \, \tilde{\theta}^\top \Delta(x) \varepsilon_{\hat{\theta}}(x) = - \upsilon(\rho) \ell(x,\pi^*_{\hat{\theta}}(x)) - \tfrac{\alpha}{\gamma} \varepsilon_{\hat{\theta}}(x)^\top \varepsilon_{\hat{\theta}}(x) \leq 0.
    \end{aligned}
\end{equation*}
Furthermore, the uniform continuity of $\dot{V}_c(t)$ can be establishing employing similar boundedness and continuity arguments as in the proof of \cref{thm:direct}.
Since the integral of $\dot{V}_c(t)$ is finite, then $x \rightarrow x_d$ as $t \rightarrow + \infty$ by Barbalat's lemma.
\end{proof}

Note that the state predictor learning rate does not need to be adjusted online to maintain stability. 
By contrast, the direct learning rate \emph{must} be adjusted to cancel the effects of learning transients on the closed-loop system.
This arises from the fundamentally different roles of direct learning and the state predictor.
The purpose of direct learning is to guarantee that the system remains stable, while the state predictor adds a form of model estimation to improve the parameter estimation transients.
As in~\citep{slotine1989composite,slotine1991applied}, the state predictor can be viewed as adding a damping term in the $\tilde{\theta}$ dynamics\footnote{Recall $\theta$ is constant so $ \ \tfrac{d}{dt}{\tilde{\theta}} =\frac{d}{dt}{\hat{\theta}}$.}
%$\tfrac{d}{dt}\tilde{\theta} = \tfrac{d}{dt}\hat{\theta} - \tfrac{d}{dt}{\theta} = \tfrac{d}{dt}\hat{\theta}$.}
\begin{equation}
\label{eq:fof}
    \dot{\tilde{\theta}} + \alpha \, [\nabla^2 \psi(\hat{\theta})]^{-1} \Delta(x) \Delta(x)^\top \, \tilde{\theta} = - \gamma\, \upsilon(\rho) [\nabla^2\psi(\hat{\theta})]^{-1} \Delta(x) \nabla_{x} V^*_{\hat{\theta}}(x),
\end{equation}
which is a stable first-order filter with the damping coefficient $\alpha \, [\nabla^2 \psi(\hat{\theta})]^{-1} \Delta(x) \Delta(x)^\top \geq 0$. 
Observe that if $\alpha \, [\nabla^2 \psi(\hat{\theta})]^{-1} \Delta(x) \Delta(x)^\top = 0$ then the right hand side of \cref{eq:fof} must also be zero. 
Also, for further conceptual clarity one could replace the left-hand side of \cref{eq:fof} by a virtual contracting system \citep{lohmiller1998contraction,jouffroy2004methodological,wang2005partial}.

\subsection{Incorporating State Constraints}
\label{sec:constraints}
Although not explicitly represented in the optimal control problem \cref{eq:ocp}, the proposed framework can incorporate state constraints by constructing an augmented stage cost $\, \ell_a : \mathbb{R}^n \times \mathbb{R}^m \rightarrow \mathbb{R}_+ \, $ that penalizes states near the constraint. 
For example, consider a constraint set $\, \mathcal{C}\, $ with boundary $\, \partial C\, $ and interior $\,\mathrm{Int} \, \mathcal{C}$.
Let the state constraint be of the form $\, g(x) \leq 0\, $ where $\, g(x) \rightarrow 0\ $ as $\, x \rightarrow \partial \mathcal{C}$.
The augmented stage cost can then be defined to be $\, \ell_a(x,\pi(x)) \triangleq \ell(x,\pi(x)) + \sigma(g(x))\, $ where $\, \sigma(\cdot)\, $ is continuously differentiable, $\, \sigma(\cdot) \geq 0\, $ for all $\, x \in \mathrm{Int}\,\mathcal{C}$, and $\, \sigma(g(x)) \rightarrow \infty \, $ as $\, x \rightarrow \partial \mathcal{C}$. 
The value function $\,V^*(x(t)) = \int_t^\infty \ell_a (x(\tau),\pi(x(\tau)))\,d\tau \,$ is still a control Lyapunov function as it possess all the necessary properties (see Proposition~\ref{prop:clf}).
In contrast to control barrier functions \cite{ames2016control,ames2019control}, using an augmented cost function yields less myopic control since the optimal policy has explicit knowledge of the constraints.

%%%%%%%%%%%%%%%%%%%%%%%%%%%%%%%%%
\section{Experiments: Mountain Car}

\begin{figure}[t]
    \centering
    \includegraphics[width=\columnwidth]{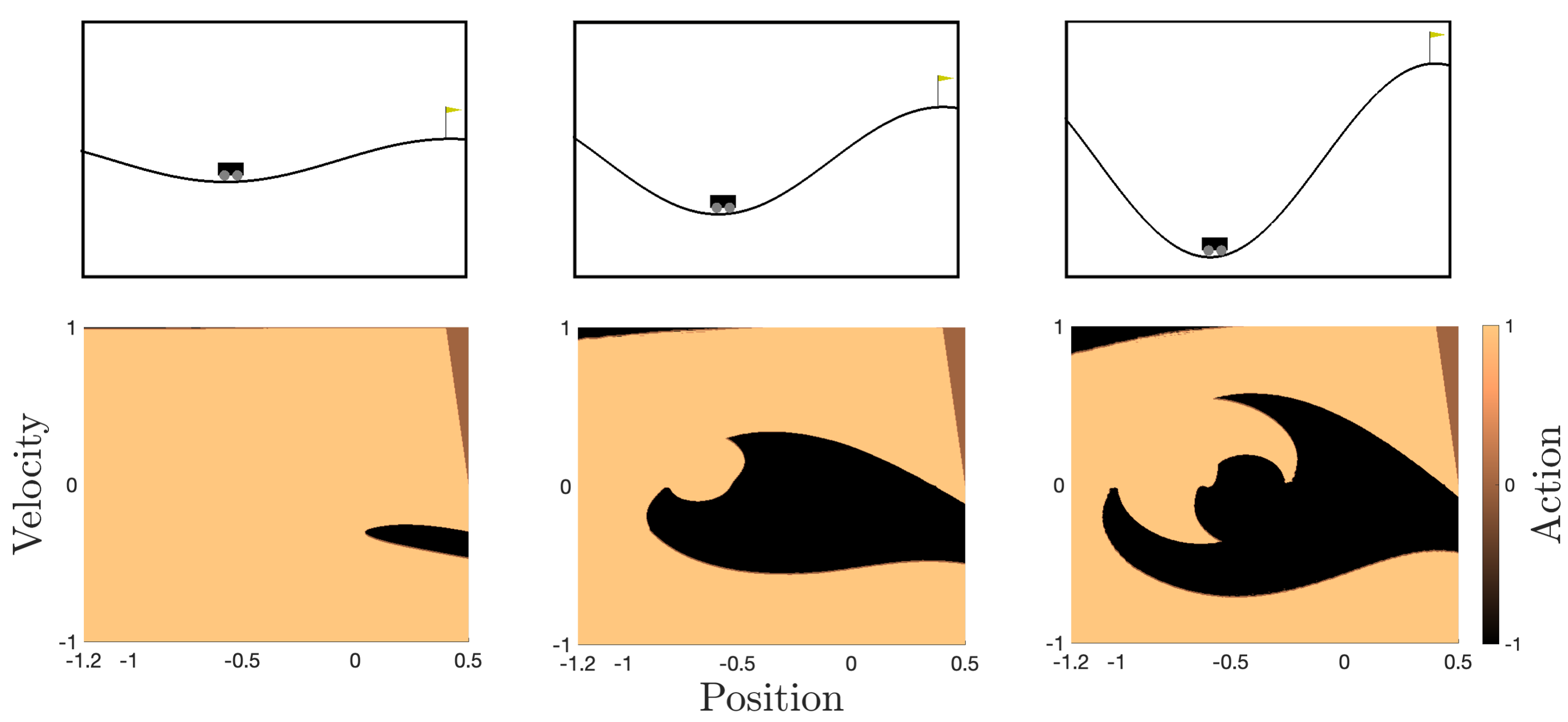}
    \vskip -0.2in
    \caption{The proposed learning algorithms were tested in different mountain car environments (top). Each environment has its own unique optimal policy (bottom).}
    \label{fig:sim}
    \vskip -0.2in
\end{figure}

\paragraph{Overview.} 
The direct and composite learning algorithms were tested on a modified version of the benchmark mountain car problem.
The goal is for an under-powered car to reach the top of a mountain with an unknown slope. 
The continuous-time optimal control problem is
\begin{equation}
    \begin{aligned}
        V_\theta^*(p(t),v(t)) & =\underset{\pi \, \in \, \Pi}{\mathrm{min}} ~  \int \limits_t^\infty \ell_a (p(\tau),v(\tau))\,d\tau \\
        \dot{p} & = v, ~ \dot{v} =  0.1 \, \pi_\theta(p,v) - \, \theta \, \mathrm{cos}(3\, p) , \\
        p & \in [-1.2,\, 0.5]\,\mathrm{m}, ~ v \in [-1,\, 1]\, \mathrm{m / s}, \\
        \pi_\theta(p,v) & \in \{-1,\,0,\,1\} \, \mathrm{m / s^2}, ~ \theta \in [0.05,\, 0.4]\, \mathrm{m / s^2}.
    \end{aligned}
\label{eq:mc_ocp}
\end{equation}
The slope of the mountain is considered to be unknown, but varies continuously from relatively flat to steep (top row of \cref{fig:sim}).
The corresponding optimal policies are significantly different depending on the slope of the mountain, as seen in the bottom row of \cref{fig:sim}.
The policy for the relatively flat mountain is to predominately drive forward as the car has enough power to reach the top.
However, as the mountain becomes steeper the policy changes to a switching function where the vehicle must accelerate forward and backward until it has enough momentum to reach the top.
The dramatically different feedback policies further motivates the use online learning to ensure the closed-loop system is stable and performing well.

\paragraph{Parameters.} All differential equations were discretized with a time step $dt = 0.001\,\mathrm{s}$.
The direct and state predictor learning rates were $\gamma = 0.02$ and $\alpha = 200000$, respectively. 
Note the large magnitude difference is due to the relative sizes of the error signal used by each component.
The stage cost was chosen to be $\,\ell(p) = 1 - \mathrm{exp}(-40|p-p_d|)\,$ which is continuously differentiable and satisfies $\ell(p_d) = 0$.
The augmented stage cost $\,\ell_a(\cdot)\,$ utilized an $L_1$ penalty on states near the path constraints of \cref{eq:mc_ocp}.
The learning rate scaling function was $\,\upsilon(\rho) = 0.9 \, e^\rho + 0.1\,$ and $\eta = 100$.
The parameter estimates were bounded by using the the Bregman divergence with the log function.

\paragraph{Methods.} Optimal value functions and feedback policies were computed offline using value iteration (VI) for several different mountain slopes. 
The dynamics were discretized with a time step of $0.1\,\mathrm{s}$ so the discrete-time version of VI could be employed. 
Bilinear interpolation between states was utilized to improve accuracy and rate of convergence.

\begin{figure}[t]
  \centering
\subfigure[Relatively flat mountain.]{\label{fig:results_shallow}\includegraphics[width=.45\columnwidth]{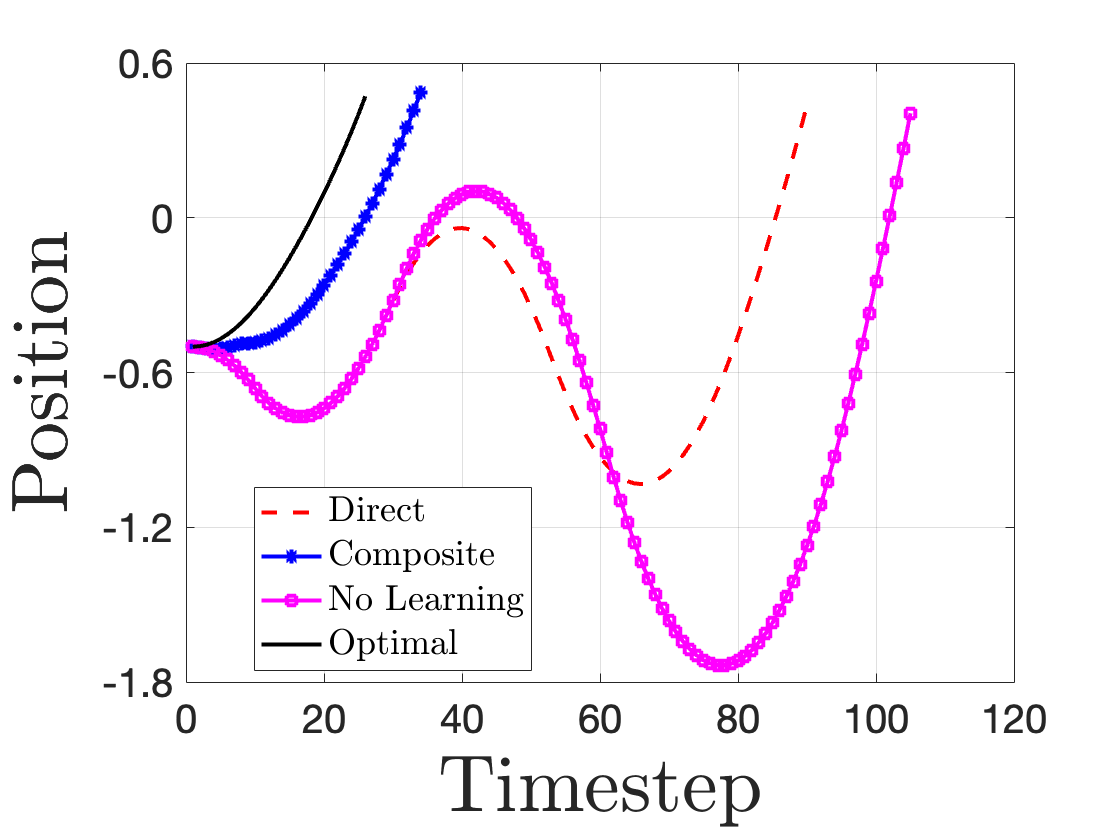}}
\hspace{1em}
\subfigure[Steep mountain.]{\label{fig:results_steep}\includegraphics[width=.45\columnwidth]{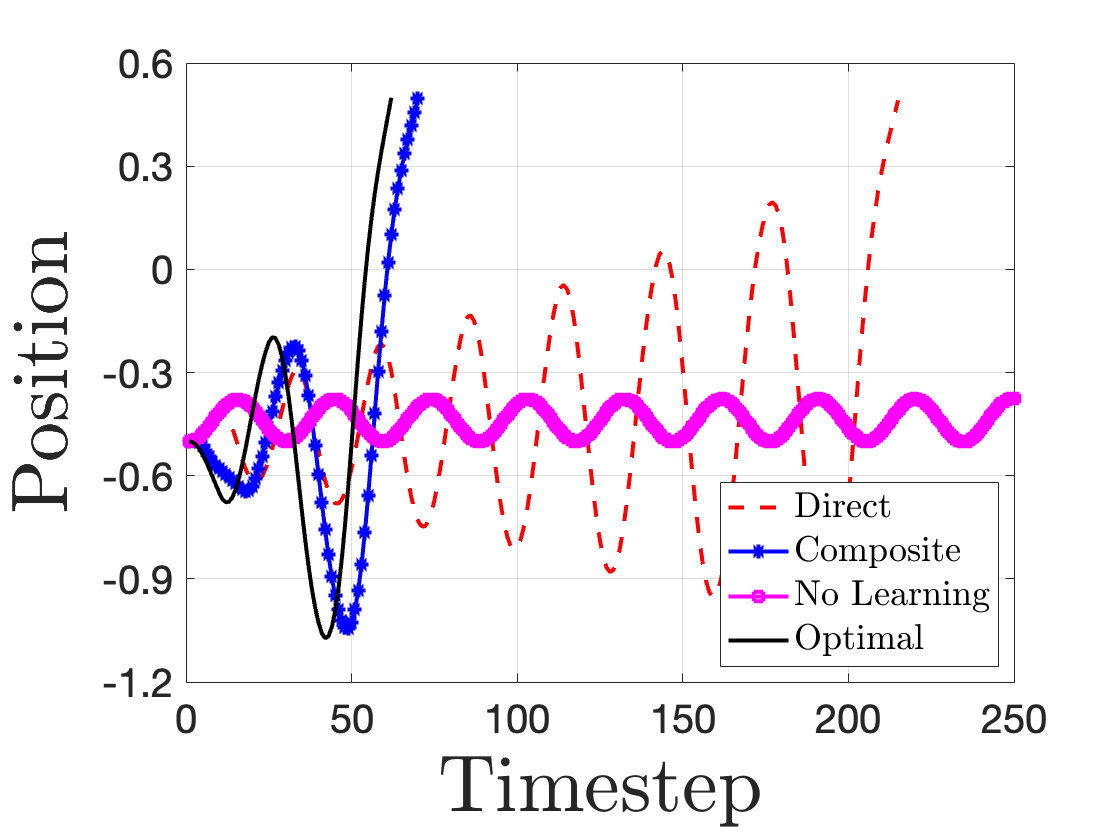}}
\caption{Comparison of four policies for different mountain car environments.}
\label{fig:results}
\vskip -0.2in
\end{figure}

\paragraph{Results.}
The direct and composite learning algorithms were evaluated in different environments by comparing their performance to the true optimal policy and to a static (no learning) policy.
Each policy (aside from the optimal) was initialized with model parameters that were furthest from that of the true model, i.e., the policy for the relatively flat mountain ($\theta = 0.05 \, \mathrm{m / s^2}$)  was used on the steep mountain ($\theta = 0.4 \, \mathrm{m / s^2}$) and vice versa. 
\cref{fig:results} shows the closed-loop position for the different policies when deployed on the relatively flat (\cref{fig:results_shallow}) and steep (\cref{fig:results_steep}) mountain.
In both environments, the direct and composite algorithms outperform the static policy.
The steep mountain test case (\cref{fig:results_steep}) shows that the closed-loop system still converges to the state of zero cost with the direct and composite learning policies despite the initial policy being unstable.
The phase portraits of each policy, in addition to the optimal value function, are also shown in \cref{fig:phase}.
Quantitatively, the composite learning policy most closely resembles the optimal policy, yielding a similar closed-loop cost as shown in \cref{table:results}.
The improved performance is an immediate benefit of incorporating state perdition in learning.
These experimental results demonstrate the strong stability guarantees the two proposed learning algorithms possess, in addition to the benefits of composite learning in achieving near-optimal closed-loop performance.

\begin{figure}[t]
    \centering
    \includegraphics[width=\columnwidth]{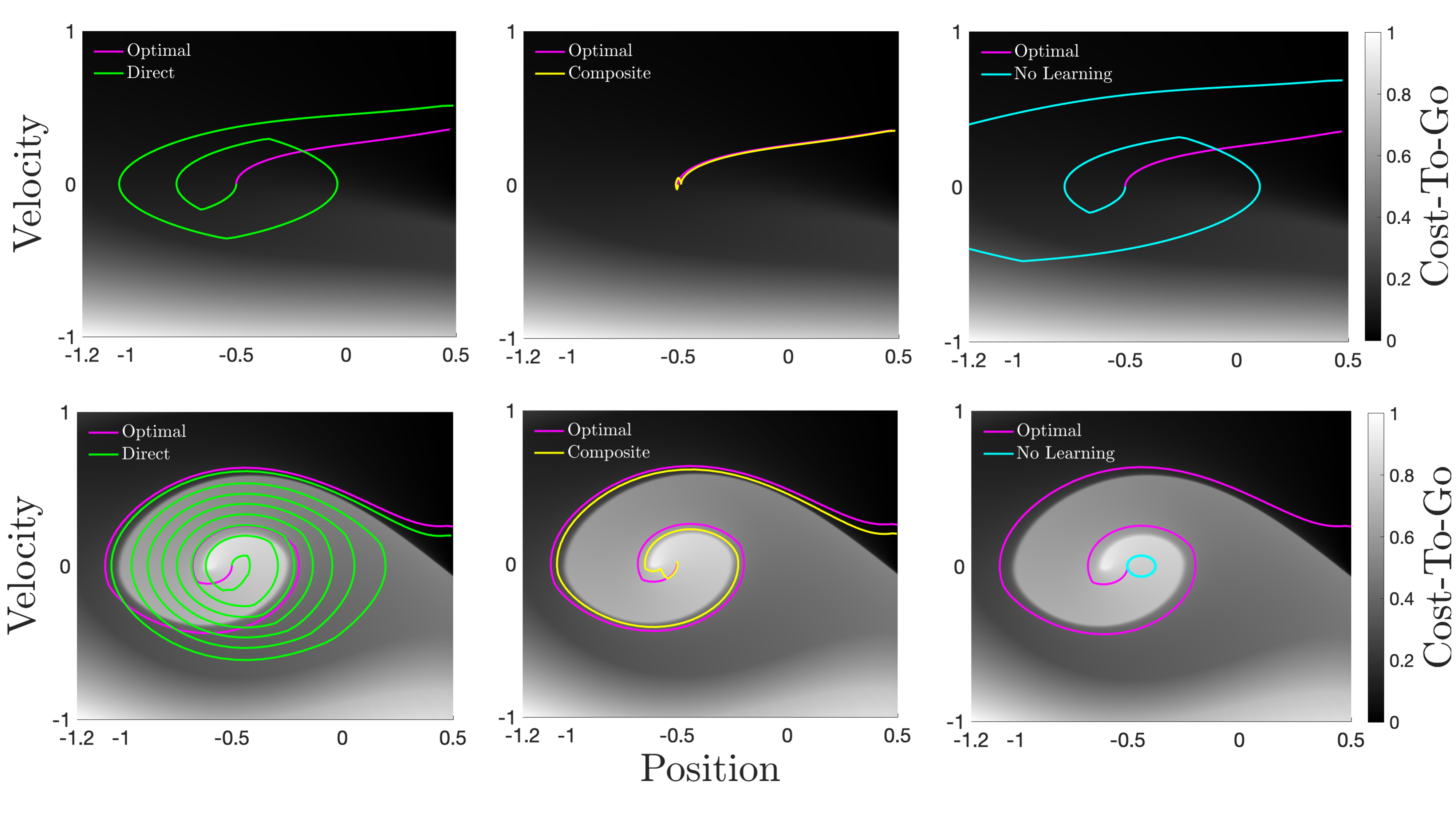}
    \vskip -0.2in
    \caption{Phase portraits of the four tested policies for the relatively flat (top) and steep (bottom) mountain environment. The normalized value function (cost-to-go) is also shown.}
    \label{fig:phase}
    \vskip -0.2in
\end{figure}

\begin{table}[!t]
    % \vskip -0.1in
    \caption{Closed-Loop Cost of Different Policies}
    % \vskip -0.1in
    \vspace{.1cm}
    \label{table:results}
    \begin{center}
    \begin{tabular}{|c|c c c c|}
    \hline 
     Mountain Grade & No Learning & Direct & Composite & Optimal \\
    \hline 
    {Relatively Flat} & 20.9 & 17.9 & \textbf{6.57} & 5.01 \\
    % {Short} & 11.0 & 7.62 & \textbf{1.39} & 1.28 \\
    % {Shallow} & 8.82 & 5.67 & \textbf{1.40} & 1.28 \\
    {Steep} & $\infty$ & 42.7 & \textbf{13.7} & 12.1 \\
    \hline
    % {Steep} & 28.0 & 18.4 & \textbf{9.56} & 9.05 \\
    %  {Steep} & 36.9 & 18.4 & \textbf{9.56} & 9.29 \\
    \end{tabular}
    \end{center}
    \vskip -0.1in
\end{table}

%%%%%%%%%%%%%%%%%%%%%%%%%%%%%%%%%
\section{Discussion}
This work proposed two online learning algorithms that can be combined with optimal feedback policies to improve closed-loop performance when the underlying dynamical model is not fully known. 
By taking a control-theoretic approach, we proved asymptotic convergence of the closed-loop system to the state of zero cost.
We also empirically showed that combining direct learning with a state predictor can yield a near-optimal stable policy.
The approach leverages the certainty equivalence principle when computing optimal policies through online adjustment of the learning rate.
Unlike other existing approaches, stability is \emph{guaranteed}.
The proposed learning framework may have a profound impact on offline computation/training of optimal policies as it is highly parallelizable, i.e., generate $N$ policies for $N$ different dynamical models independently.
More generally, it bypasses the ``one-size-fits-all" strategy of finding a single policy for all possible models,  simplifying the control synthesis problem.

Several future works are of interest.
The first is a thorough regret analysis \citep{hazan2007logarithmic,agarwal2019logarithmic,dean2018regret,boffi2021regret} of the proposed approach. 
The empirical results presented here hint that the composite adaptive policy is nearly optimal, but a rigorous analysis is needed before any claims can be made. 
Moreover, regret may play a fundamental role in designing the state predictor or other modifications to the proposed learning law.
Nonetheless, our approach possess strong stability characteristics that apply to both linear and nonlinear systems in its current form.
Investigating the relationship between the choice of the stage cost, i.e., reward shaping \citep{ng1999policy}, and robustness is also of interest. 
As shown in the Appendix, an appropriate choice of the stage cost can guarantee convergence to the state of zero cost even without online learning.
However, this strategy alone will inherently be conservative so a combination of reward shaping for robustness and online learning is ideal to maximize closed-loop performance. 
% As shown in the Appendix of \citep{lopez2021adaptive}, an appropriate choice of the stage cost can guarantee convergence to the state of zero cost even without online learning, although an adaptive robust strategy is, in our view, the most effective. 
More analysis on adaptive combinations of optimal policies is also of interest as several phenomenon encountered in nature are interconnected combinations of systems that minimize some cost function, e.g., energy.
Lastly, additional empirical evaluation will be pursued, including investigating uses in RL, differential games, transfer learning, and model predictive control.

%%%%%%%%%%%%%%%%%%%%%%%%%%%%%%%%%
\section*{Appendix}
\label{sec:appendix}
\begin{lemma}[Barbalat]
If a function $g(t)$ is uniformly continuous and $\, \lim_{t\rightarrow \infty}\int_0^t g(\tau)\, d\tau < \infty$, then $\, \lim_{t\rightarrow \infty} g(t) = 0$. 
\end{lemma}

\paragraph{Reward Shaping} 
If the optimal value function $\,V_{\bar{\theta}}^*(x)\,$ is computed with nominal model parameters $\bar{\theta}$, then for the actual system
\begin{equation*}
    \begin{aligned}
        \dot{V}_{\bar{\theta}}^*(x) & = \nabla_x V_{\bar{\theta}}^* (x)^\top  \left( f(x,\pi_{\bar{\theta}}^*(x)) - \Delta(x)^\top \bar{\theta} \right) + \nabla_x V_{\bar{\theta}}^* (x)^\top \Delta(x)^\top \tilde{\theta} \\
        & = -\ell(x,\pi_{\bar{\theta}}^*(x)) + \nabla_x V_{\bar{\theta}}^* (x)^\top \Delta(x)^\top \tilde{\theta} \\
        & \leq -\ell(x,\pi_{\bar{\theta}}^*(x)) + \left|\nabla_x V_{\bar{\theta}}^* (x)^\top \Delta(x)^\top\right| \|\tilde{\theta}\|_\infty,
    \end{aligned}
\end{equation*}
where $\, \tilde{\theta} \triangleq \bar{\theta} - \theta$.
In order for $\, \dot{V}_{\bar{\theta}}^*(x) \leq 0\, $ then $\, \ell(x,\pi_{\bar{\theta}}^*(x)) \geq \left|\nabla_x V_{\bar{\theta}}^* (x)^\top \Delta(x)^\top\right| \|\tilde{\theta}\|_\infty$, i.e., $\ell(\cdot)$ must dominate the model error term.
Stability can then be guaranteed despite the presence of model error with the appropriate selection of the cost function.

%%%%%%%%%%%%%%%%%%%%%%%%%%%%%%%%%
\paragraph{Acknowledgments:}{We thank Michael Everett for stimulating discussions and Kenny Chen for helping with the figures.}

% \newpage

\bibliography{ref}

\end{document}